\documentclass[preprint,12pt,authoryear]{elsarticle}

\usepackage{amsthm}
\usepackage{xcolor}
\usepackage{mathrsfs}
\usepackage{lipsum}
\usepackage{amsfonts}
\usepackage{algorithmic}
\usepackage{amssymb}
\usepackage{amsmath}
\usepackage{graphicx}
\usepackage{subfig}
\usepackage{subcaption}
\usepackage{epstopdf}
\usepackage{mathtools}
\usepackage{tabularx}
\usepackage{algorithm2e}
\usepackage{wrapfig}
\usepackage{hyperref}
\usepackage{cleveref}
\usepackage{bm}
\usepackage{amsopn}
\usepackage{sidecap}

\SetKwComment{Comment}{/* }{ */}

\ifpdf
\DeclareGraphicsExtensions{.eps,.pdf,.png,.jpg}
\else
\DeclareGraphicsExtensions{.eps}
\fi
\newcommand{\RN}[1]{%
  \textup{\uppercase\expandafter{\romannumeral#1}}%
}
\newcommand{\rn}[1]{%
  \textup{\expandafter{\romannumeral#1}}%
}
\usepackage{tikz}

\newtheorem{theorem}{Theorem}
\newtheorem*{lemma*}{Lemma}
\newtheorem*{corollary*}{Corollary}
\newtheorem{remark}{Remark}

\theoremstyle{definition}

\newcommand{\dee}{\partial}

\journal{Journal of Theoretical Biology}

\begin{document}

\begin{frontmatter}



\title{Inferring Exocytosis Profiles from Cell Shapes Using a Dual-Configuration Model of Walled Cell Tip Growth}

%
%

\author[add3]{Kamryn Spinelli}

\author[add1]{Chaozhen Wei}
\author[add4,add5]{Luis Vidali}

\author[add1,add5]{Min Wu}\corref{Email address}
\ead{englier@gmail.com}

\address[add3]{Department of Mathematics, Brandeis University}
\address[add1]{Department of Mathematical Sciences, Worcester Polytechnic Institute}
\address[add4]{Department of Biology and Biotechnology, Worcester Polytechnic Institute}
\address[add5]{Bioinformatics and Computational Biology Program, Worcester Polytechnic Institute}

\cortext[CorrespondingAuthor]{Corresponding author}


\begin{abstract}
Tip growth in filamentous cells, such as root hairs, moss protonemata, and fungal hyphae, depends on coordinated cell wall extension driven by turgor pressure, wall mechanics, and exocytosis. We introduce a dual-configuration model that incorporates both turgid and unturgid states to describe cell wall growth as the combined effect of elastic deformation and irreversible extension. This framework infers exocytosis profiles directly from cell morphology and elastic stretches, formulated as an initial value problem based on the self-similarity condition. Applying the model to {\it Medicago truncatula} root hairs, moss {\it Physcomitrium patens} protonemata, and hyphoid-like shapes, we find that exocytosis peaks at the tip in tapered cells but shifts to an annular region away from the apex in flatter-tip cells beyond a threshold. The model generalizes previous fluid models and provides a mechanistic link between exocytosis distribution and cell shape, explaining observed variations in tip-growing cells across species.
\end{abstract}

\begin{highlights}
\item A dual-configuration model describes tip growth as a combination of elastic deformation and irreversible extension.
\item The model generalizes fluid-based approaches and explains exocytosis pattern variability across species.
\item It provides a proof of concept for predicting exocytosis profiles from elastic strains and cell configurations.
\item Exocytosis peaks at the apex in tapered cells and shifts to an annular region in flatter-tip cells.
\end{highlights}

\begin{keyword}
cell wall  \sep tip growth  \sep morphogenesis \sep elasticity \sep Lockhart extensibility


\end{keyword}

\end{frontmatter}

\section{Introduction}\label{sec:intro}
 In plants and fungi, the cell wall boundary defines the cell morphology and sustains the internal turgor pressure resulting from osmosis \cite{cosgrove-2005}. In tip-growing filamentous cells, such as in root hairs \cite{galway2006root}, moss protonemata \cite{bibeau2021quantitative}, pollen tubes \cite{hepler2013control}, and fungal hyphae \cite{bartnicki2000mapping}, which are one cell wide, cell walls need to extend along one direction to optimize the migration speed. How this growth mode is regulated is of critical importance for plant and fungus development and survival.

In the absence of external environmental perturbations, the filamentous cell wall in its apical domain remains approximately axisymmetric and extends along its axis of symmetry \cite{dumais-2004,bibeau2021quantitative,bartnicki2000mapping}. Many filamentous cells elongate while preserving their apical domain shape, a phenomenon known as self-similar tip growth \cite{self-similar-tip-growth-goriely-tabor}. For over a century, researchers have tracked the relative cell-wall surface observable strain rates $\dot{\epsilon}_s$ and $\dot{\epsilon}_\theta$ along the meridional and circumferential directions, respectively, by measuring the rate of tangential (meridional) velocity $v_s=\partial s/\partial t$ and local cell radius $r$ via \cite{Reinhardt,castle1958topography,shaw2000cell,dumais-2004} (see Fig.\ref{fig:growth-components}a):\begin{eqnarray} 
\label{dotrate} 
\dot{\epsilon}_s=\frac{\partial v_s}{\partial s} \text{ } \text{ and }\dot{\epsilon}_\theta=\frac{v_s}{r}\frac{\partial r}{\partial s} \end{eqnarray} where $s$ is the meridional distance from the tip. In many systems, such as {\it Chara} rhizoids \cite{hejnowicz-1977} and fission yeast \cite{fission-yeast}, experimental tracking of wall surface expansion shows that the highest strain rates, $\dot{\epsilon}_s$ and $\dot{\epsilon}_\theta$, occur at the tip. However, in some systems, such as root hair growth, the highest  strain rates are observed in an annular region around the tip \cite{shaw2000cell,dumais-2004}. It is not clear what factors are responsible for these qualitatively different wall extension profiles. 
   \begin{figure}[h]
    \centering
    \includegraphics[width=0.8\linewidth]{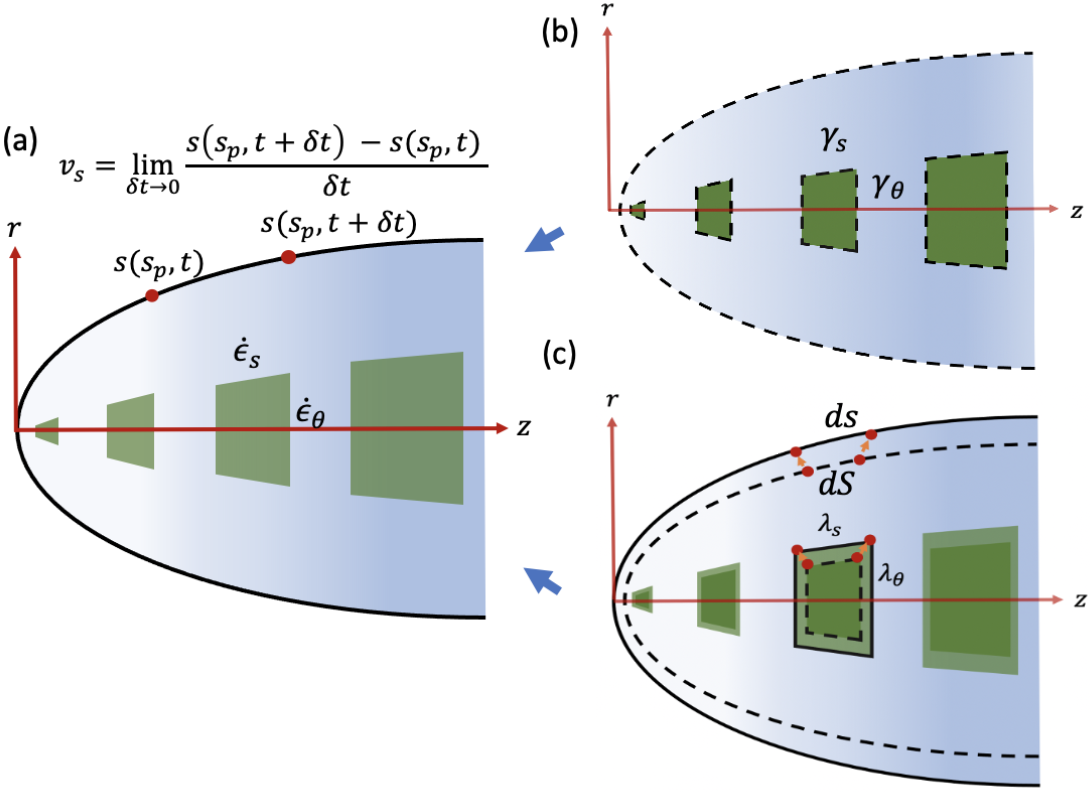}
    \caption{ Illustration of the dual-configuration model. 
(a) The observed surface deformation (in the turgid configuration) reflects both irreversible wall surface-area growth and reversible elastic stretching. 
(b) We introduce a relaxed (unturgid) configuration to represent the reference surface resulting from wall surface-area growth alone, prior to elastic loading. 
(c) Turgor pressure imposes elastic deformation on this relaxed surface to yield the observed (turgid) geometry. 
The local strain rates, \( \dot{\epsilon}_s \) and \( \dot{\epsilon}_\theta \), observed on the turgid wall are thus interpreted as the result of combined contributions from growth-driven deformation of the relaxed configuration—quantified by \( \gamma_s \) and \( \gamma_\theta \)—and elastic stretch ratios  due to turgor, given by \( \lambda_s \) and \( \lambda_\theta \). The origin of the r-axis is fixed at the centerline of the cell, while the origin of the z-axis is arbitrary and does not affect the cell configuration.}
    \label{fig:growth-components}
\end{figure}

Wall extension is governed by multiple interconnected processes \cite{dumais2021mechanics}. First, turgor pressure induces tension on the cell wall surface, providing mechanical drive or cue for expansion. Second, the wall matrix undergoes local network rearrangements in response to these tension, facilitated by wall-loosening chemical factors \cite{cosgrove2000loosening,cosgrove-2005}. In fast-growing pollen tubes, wall loosening is associated with the turnover and remodeling of crosslinks in the polymeric network, involving enzymes such as pectin methylesterases \cite{rojas-pollen-tube}. In slower-growing root hairs and moss protonemata, wall loosening is thought to involve expansins, which likely disrupt noncovalent bonding between laterally aligned polysaccharides such as cellulose \cite{cosgrove2024plant}. Third, new wall materials and wall-loosening factors are delivered via exocytosis. In root hairs and moss protonemata, some wall components—particularly cellulose—are also synthesized directly at the plasma membrane by cellulose synthase complexes. However, without concurrent delivery of wall-loosening agents via exocytosis, such synthesis may lead primarily to wall thickening rather than expansion.

Many previous theoretical approaches are based on Lockhart’s wall extensibility model \cite{Lockhart1965}, which directly links the observable  strain rates ($\dot{\epsilon}_s$ and $\dot{\epsilon}_\theta$) to cell wall mechanics. Consequently, the spatial patterning of ``wall extensibility" has been used to explain variations in cell morphology during self-similar growth. However, these approaches often rely on simulating cell wall flow until the cell shape reaches a steady state,  and how the wall elastic strains respond to the change of turgor pressure is not captured. Additionally, they are typically coupled with wall-extension tracking experiments, which require real-time tracking of material points on the cell wall. 

Recent advances in experimental techniques \cite{fission-yeast,davi-2019,deng2022inferring,chevalier2023cell} have enabled the measurement of elastic strains and  elastic moduli in tip-growing cell walls by tracking cell deformation between turgid and unturgid configurations at the tip and lateral body regions. These studies reveal that solid properties adapt to cell width in fission yeast \cite{davi-2019} and can exhibit significant gradients in hyphae \cite{chevalier2023cell}. Thus, the understanding of the distribution of elastic properties of the wall during wall expansion is timely.

Notably, these experiments are faster than wall-flow tracking methods, as measuring elastic deformation takes only a few seconds. Since these experiments involve two distinct cell wall configurations (turgid and unturgid), a modeling approach that explicitly incorporates both cell states can provide new insights into tip growth.  In addition, capturing two configurations clearly define which part of the wall expansion is due to elastic deformation, thus reversible, and which part is the irreversible growth due to material addition and remodeling. Despite the added complexity, such an approach holds promise for bridging experimental observations with theoretical predictions.

In this study, we introduce a dual-configuration tip growth model that describes cell wall growth as the combined effect of elastic deformation and irreversible extension, driven by local exocytosis and elastic strains.  This approach shares the same spirit as the theory of growth-elasticity developed over the past three decades \cite{rodriguez1994stress,skalak1996compatibility}, in that it considers two configurations at any given time: a grown configuration prior to elastic deformation and a current configuration after the elastic deformation. However, it also differs in a fundamental and more straightforward way. In our framework, the two global configurations are well-defined and distinct, whereas in classical growth-elasticity theory, the intermediate grown (irreversibly changed) configuration is typically defined only locally via growth tensors. A global representation of this configuration is generally intractable—hence the use of terms such as incompatible growth and virtual configuration. 

Leveraging the self-similarity condition \cite{hejnowicz-1977}, we formulate a well-defined initial value problem to infer the exocytosis profile from cell shape data and apply the framework across multiple cell species. Furthermore, we connect the dual-configuration model with previous models. We demonstrate that the two-configuration model generalizes previous fluid models, which emerge as approximations when elastic deformation is small. We apply the model to cell outline data of {\it Medicago truncatula} root hairs and moss {\it Physcomitrium patens} protonemata, predicting that the exocytosis profile peaks in an annular region when the cell tip becomes less tapered beyond a critical threshold.

\section{The model}
Treating the cell wall as a thin shell of revolution generated by its meridian, we assume that at any time point $t$, one material point $s_p$ along the meridian is mapped to two configurations, one turgid state $\left(s(s_p,t),r(s_p,t)\right)$ and one unturgid state $\left(S(s_p,t),{R}(s_p,t)\right)$ where $s$ and $S$ are the meridional distances from the tip end and $r$ and ${R}$ are the local cell  radius in the turgid and unturgid states, respectively.
 We describe the growth of the unturgid surface path by 
\begin{eqnarray}
    \label{eq1}
    \frac{\dee S}{\dee t}\bigg|_{s_p} = \int_0^{S} \gamma_s\ dw \text{ } \text{ and } \text{ } \frac{\dee {R}}{\dee t}\bigg|_{s_p} = \gamma_\theta {R}
\end{eqnarray}
where $\gamma_s$ and $\gamma_\theta$,  with units of inverse time, are the relative rates of extension of the unturgid cell wall surface in the meridional and circumferential directions respectively (Fig. \ref{fig:growth-components}b). Henceforth, we name them growth strain rates in the unturgid configuration.  We note that by defining the tangential (meridional) velocity $V_S =  \frac{\dee S}{\dee t}|_{s_p}$ in the unturgid configuration, the first part of the equation follows by integrating the relation $\gamma_s = \frac{\partial V_S}{\partial S}$ from the pole to the material point $s_p$ located at arclength 
$S$. The second part of the equation follows by combining the relation $\gamma_\theta = \frac{V_S}{ R}\frac{\partial R}{\partial S}=\frac{1}{R}\frac{\partial S}{ \partial t}\frac{\partial R}{\partial S}=\frac{1}{ R}\frac{\partial R}{\partial t}$.


The unturgid and turgid configurations are connected through the elastic stretches along the meridional and circumferential directions, respectively (Fig.~\ref{fig:growth-components}c). These stretches are defined as:  
\begin{eqnarray}
\label{stretches}
\lambda_s = \frac{\partial s}{\partial {S}}, \quad \lambda_\theta = \frac{r}{{R}},
\end{eqnarray}
where $\lambda_s(S)$ represents the meridional stretch ratio and $\lambda_\theta(S)$ represents the circumferential stretch ratio.
That said, we describe the cell wall surface area growth coming from two contributions: 1)  growth strain rates of the unturgid surface patch $\gamma_s$ and $\gamma_\theta$; and 2) the stretch ratio due to elastic deformations $\lambda_s$ and $\lambda_\theta$. We note that in \cite{fission-yeast}, two surface configurations are introduced implicitly by connecting the meridional flow fields
$ v_s := \frac{\partial s}{\partial t} |_{s_p}$ and $V_{S} := \frac{\partial {S}}{\partial t} |_{s_p}$, 
in the turgid and unturgid configurations, respectively. However, the relation $v_s = \lambda_s V_{S}$ from \cite{fission-yeast} is valid only in the self-similar growth regime. Given the model from \cite{fission-yeast} is based on flow descriptions, we denote it as a dual-flow model.  Our model provides a more general formulation by explicitly modeling the coordinates themselves rather than their time derivatives. For a detailed comparison, see \ref{Appendix_B}.

 \subsection{Cell wall mechanics}
 
To describe the cell wall mechanics, both wall-surface tension and bending moment may be considered \cite{boudaoud2003growth,fission-yeast}. For simplicity, we suppose that the force $\vec{F}$ on the cell wall is dominated by the wall-surface tensions  $\sigma_s(S)$ and $\sigma_\theta(S)$ —with units of force per length— along meridional and circumferential direction respectively, and the turgor pressure $P$: 
\begin{eqnarray} \label{eqn:local-force-balance}
 \vec{F}:=\frac{\partial (\sigma_s \hat{t}) }{\partial s} + \frac{\sigma_s - \sigma_\theta}{r} \hat{r} - \frac{\sigma_s \sin\alpha}{r} \hat{n} + P \hat{n}=0,
\end{eqnarray}
where $\hat{r}$ is the unit vector along the cell radius, $\hat{t}$ is the local unit tangent, $\hat{n}$ is the local unit normal vector towards the cell exterior, and $\alpha$ is the angle from $\hat{t}$ to $\hat{r}$.  We derive this vector equation from $\nabla_s\cdot\boldsymbol\sigma +P\hat{n}=0$ \cite{pozrikidis1992boundary} where $\nabla_s = \hat{t} \frac{\dee}{\dee s} + \frac{\hat{\theta}}{r} \frac{\dee}{\dee \theta}$ and $\boldsymbol\sigma = \sigma_s \hat{t} \otimes \hat{t} + \sigma_\theta \hat{\theta} \otimes \hat{\theta}$, $\theta$ being the azimuth with respect to the axis of symmetry. See \ref{Appendix_force} for the detail. One can check that the widely used Young-Laplace Law \cite{dumais-2006,Camps2009,Camps2012,rojas-pollen-tube,deJong2020,goriely-tabor-jtb,self-similar-tip-growth-goriely-tabor,Fayant2010}: $\kappa_s \sigma_s + \kappa_\theta \sigma_\theta = P$ and $\kappa_\theta \sigma_s = P/2$ (where $\kappa_s$  and $\kappa_\theta$ —with units of inverse length—  are the curvatures along the meridional and circumferential directions, respectively) can be derived from this equation.

To relate the wall-surface tension with the elastic stretches, we first consider the nonlinear elastic theories \cite{self-similar-tip-growth-goriely-tabor}. In particular, we consider the constitutive law
\begin{eqnarray}
\sigma_{s} &=& \frac12\mu_h(1/(\lambda_{\theta})^2-1/(\lambda_{s})^2)+K_h(\lambda_{s}\lambda_{\theta}-1)\label{rel1}\\
\sigma_{\theta} &=& \frac12\mu_h(1/(\lambda_{s})^2-1/(\lambda_{\theta})^2)+K_h(\lambda_{s}\lambda_{\theta}-1)\label{rel2}
\end{eqnarray}
where $K_h=h\times K$ and $\mu_h=h\times\mu$ are the rescaled bulk modulus and shear modulus, respectively, $h$ being the cell wall thickness in the thin-shell approximation. $K_h$ and $\mu_h$ have units of force per length, whereas $K$ and $\mu$ have units of force per unit area (i.e., force per length squared).  In addition, we will also consider the linear law from \cite{fission-yeast}:
\begin{eqnarray}
\label{linear_law}
(\lambda_s - 1) = \frac{\sigma_s - \nu \sigma_\theta}{E_h}, \quad (\lambda_\theta - 1) = \frac{\sigma_\theta - \nu \sigma_s}{E_h}.
\end{eqnarray}
where $E_h = h \times E$ and $\nu$ are the rescaled Young’s modulus and Poisson’s ratio, respectively. These two laws are equivalent up to the first order of small strains, and $\mu_h=\frac{E_h}{2(1+\nu)}$  and $K_h=\frac{E_h}{2(1-\nu)}$.

At the tip end of the apical domain, we have the condition
  \begin{eqnarray}
    \label{bc1}
r|_{{S}=0}=0  \text{ } \text{ and } \text{ }  \vec{F}\cdot\hat{z}|_{{S}=0}=0
\end{eqnarray}
due to fixing the tip along $r$-direction and force-balance along the long axis $z$-direction $\hat{z}$. Similarly, at the shank boundary ${S}=s^b$, we have 
  \begin{eqnarray}
  \label{bc2}
z|_{{S}=s^b}=z^b  \text{ } \text{ and } \text{ }  \vec{F}\cdot\hat{r}|_{{S}=s^b}=0
\end{eqnarray}
due to no displacement of the shank boundary along the $z-$direction and force-balance along the $r-$direction.  The quantity \( z^b \) serves only as a reference position and does not affect the cell configuration. Notice Eq.(\ref{stretches})-(\ref{bc2}) is an elastostatic problem even if the tip growth has not settled in the self-similar growth state. It only involves solving the turgid cell configuration given the unturgid one.

The essential difference between our work and previous modeling approaches (summarized in Sec.~\ref{comparision_sec}) is that elastic stretch is explicitly modeled based on the two configurations and can be solved with or without wall extension. This is not the case in previous models, where the formulation requires additional modifications for model closure when wall extension ceases (i.e., no flow). Moreover, our model is compatible with recent experimental techniques \cite{davi-2019,deng2022inferring,chevalier2023cell}, which enable the measurement of elastic moduli in tip-growing cell walls by tracking cell deformation between turgid and unturgid configurations.

\subsection{The local growth strain rates of the unturgid configuration}
We assume that the local meridional and circumferential growth strain rates, $\gamma_s(s_p,t)$ and $\gamma_\theta(s_p,t)$ in the unturgid configuration, are proportional to the corresponding elastic strains, $(\lambda_s - 1)$ and $(\lambda_\theta - 1)$, consistent with previous works \cite{fission-yeast}. Additionally, we assume that the growth strain rates vanish wherever the elastic strain is zero or negative. Lastly, they are assumed to be proportional to a non-dimensional chemical concentration field, $\gamma$—that is, a concentration normalized by a reference level—which influences cell wall extension. In \cite{rojas-pollen-tube}, $\gamma$ represents the concentration of de-esterified free pectin residues. More recently, it has been associated with the exocytosis profile \cite{fission-yeast,ohairwe2024fitness}. Henceforth, we consider $\gamma(S)$ as the non-dimensional exocytosis profile.  Combining these assumptions, we obtain:
\begin{eqnarray}
\label{strain-promoted}
   \gamma_s =\beta \cdot \gamma\cdot (\lambda_s-1) \cdot 1_{\{\lambda_s>1\}}, \quad  
   \gamma_\theta =\beta \cdot \gamma\cdot (\lambda_\theta-1) \cdot 1_{\{\lambda_\theta>1\}},
\end{eqnarray}
where $1_{\chi}$ is the characteristic function of the set $\chi$,  defined by
\[
{1}_\chi(x) = 
\begin{cases}
1, & \text{if } x \in \chi, \\
0, & \text{otherwise}.
\end{cases}
\]

The parameter $\beta$—in the unit of inverse time—  determines the baseline rate, while the exocytosis profile $\gamma$ modulates the local growth strain rate in conjunction with the elastic strains.  
 The key difference from previous works is that $\gamma_s$ and $\gamma_\theta$ are defined in the unturgid configuration rather than the turgid one. While $\gamma_s$ and $\gamma_\theta$ cannot be directly tracked from the turgid configuration in a single experiment, they can be inferred in reproducible self-similar growth regimes by combining bead tracking and osmotic shock measurements across samples. More importantly, defining these quantities in the unturgid configuration offers a clear separation between irreversible wall growth and reversible elastic deformation caused by turgor pressure—a distinction that becomes particularly important near the tip, where elastic stretches are spatially non-uniform and may obscure the underlying growth pattern.

\subsection{Self-similar growth and the spatial profile of $\gamma$}
In self-similar tip growth (the apical-domain shape is in a steady state), the local growth strain rates of the material points need to satisfy a constraint  $  \frac{\partial {R}}{\partial t}\big|_{s_p} /\frac{\partial {S}}{\partial t }\big|_{s_p} ={d{R}}/{d{S}}$ where ${R}({S})$ describes the steady-state outline of the unturgid configuration. Notably, this is a purely geometric constraint: it ensures that the rate of dimensional change is compatible with the local geometry of the surface, under the assumption that expansion occurs only tangentially—i.e., material points move along the surface, and the surface itself does not shift in the normal direction.  We name this constraint the self-similarity condition, which has been derived previously in \cite{hejnowicz-1977},  where the relation is derived in the mind of turgid configuration. Here we apply it on the unturgid configuration, and  with Eqs.(\ref{eq1}) and (\ref{strain-promoted}), we derive
\begin{eqnarray} \label{eqn:self-similarity-condition}
    { \gamma(S)} = \frac{d{R} / d{S}}{{R}(\lambda_\theta - 1)\cdot1_{\{\lambda_\theta>1\}}}{\int_0^{{S}} (\lambda_s - 1)\cdot1_{\{\lambda_s>1\}} \gamma(w)\ dw}
\end{eqnarray}
which enables us to solve $\gamma$ given the steady-state cell geometry ${R}({S})$ and elastic stretch ratios $\lambda_s$ and $\lambda_\theta$. Notice that if the unturgid configuration is in a self-similar state and the elastic stretches are in steady state (i.e., constant in time but potentially varying in space), then the turgid configuration is also in a self-similar state:
\begin{equation*}
\frac{\left.\frac{\partial r}{\partial t}\right|_{s_p}}{\left.\frac{\partial s}{\partial t}\right|_{s_p}} = \frac{\lambda_\theta \left.\frac{\partial R}{\partial t}\right|_{s_p}}{\lambda_s \left.\frac{\partial S}{\partial t}\right|_{s_p}} = \left( \frac{\lambda_\theta}{\lambda_s} \right) \left( \frac{dR}{dS} \right) = \frac{dr}{ds}.
\end{equation*}

That is, the compatibility of local strain rates with the geometry is preserved under elastic deformation when the stretches are time-invariant. 
 
Equation~(\ref{eqn:self-similarity-condition}) is difficult to analyze directly due to the presence of an integral and the degeneracy at \( S = 0 \), where it reduces to the trivial identity \( \gamma(0) = \gamma(0) \). However, if we assume that \( \lambda_s > 1 \), \( \lambda_\theta > 1 \) (which removes the characteristic functions), and that \( \gamma(S) \) is smooth, then we can differentiate both sides of the equation with respect to \( S \) to obtain an expression of the form
\[
\gamma'(S) = f(S)\gamma(S).
\]
This reveals that \( \gamma(S) \) is uniquely determined as the solution to a first-order initial value problem, with the initial condition \( \gamma(0) = \overline{\gamma} > 0 \), which sets the scale of the exocytosis profile.
 See $f({S})$ and the detailed analysis in \ref{Appendix_A}.  Importantly, one can show that $ \lim_{{S} \to 0+} f({S}) = (\lambda_s'(0) - 2 \lambda_\theta'(0))/({\lambda_\theta-1})$ and prove the following Theorem. 

\begin{theorem} \label{thm:local extrema}
    Given $\lambda_\theta$, $\lambda_\theta'$,  $\lambda_s'$, $\gamma$ continuous and $\lambda_\theta(0)>1$ and $\gamma(0)>0$, the secretion distribution $\gamma$ has a local maximum at the tip ${S}=0$ if
    \begin{eqnarray}
      {\lambda_s'(0) - 2 \lambda_\theta'(0)}<0
    \end{eqnarray}
   and has a local minimum at the tip point if 
   \begin{eqnarray}
      {\lambda_s'(0) - 2 \lambda_\theta'(0)}>0.
    \end{eqnarray}
\end{theorem}

See proof in \ref{Appendix_A}. This shows that there is a local maximum of $\gamma$ at the tip location ${S}=0$ only when the elastic stretches satisfy the relation $\lambda_s'(0)- 2 \lambda_\theta'(0)<0$, and vice versa. This alone predicts that if one can measure the distribution of $\lambda_s$ and $\lambda_\theta$ near the tip by deflating the cell (e.g., through an osmoticum change), cells with $\lambda_s'(0) - 2 \lambda_\theta'(0) > 0$ may exhibit exocytosis peaking in an annular region around the tip rather than at the tip itself.  

\begin{remark}
Theorem 1 is derived from Eq. (\ref{eqn:self-similarity-condition}) which only concerns the geometry of the cell outline and the elastic strains. Thus it is valid independent of the distribution of the elastic properties and thickness of the wall varies along the meridian (i.e., $\mu$, $K$, and $h$). 
\end{remark}

\subsection{Model summary and the quasi-static inverse problem}
By rescaling variables as follows: $t \to \beta t$, ${R} \to {R} / L$, $r \to r / L$, ${S} \to {S} / L$, $s \to s / L$, $\mu_h \to \mu_h / (P L)$, $K_h \to K_h / (P L)$, and $E_h \to E_h / (P L)$—where $L = {R}(s^b, 0)$ represents the cell radius at the shank boundary—and using the same notation for the non-dimensionalized quantities, the model reduces to only two material parameters: $\mu_h$ and $K_h$ (or equivalently, $E_h$ and $\nu$ if the linear elasticity Eq.(\ref{linear_law})). The model describes dimensionless cell shape configurations, consisting of one turgid state, $\left(s(s_p, t), r(s_p, t)\right)$, and one unturgid state, $\left({S}(s_p, t), {R}(s_p, t)\right)$. 

Thus, tip growth can be simulated using the pattern of $\gamma$ as a model input by solving the growth Eqs.~(\ref{eq1}) and (\ref{strain-promoted}) together with the elastostatic system Eqs. ~(\ref{stretches})–(\ref{bc2}). For details on the solution procedure, see \ref{Appendix_C}. Rather than connecting cell shape to the pattern of $\gamma$ through simulations, this work primarily focuses on demonstrating how to infer $\gamma$ directly from cell morphology and elastic stretches via Eq.(\ref{eqn:self-similarity-condition}). Nevertheless, we briefly present simulation results for model and method validation purposes (Fig.~\ref{fig:25}). This quasi-static approach is more cost-efficient both theoretically and experimentally, as it requires only short-term tracking of elastic deformations instead of full-scale simulations or velocity-extension profiling using experimental data via Eq.~(\ref{dotrate}). 

\section{Results}

\subsection{Inference of $\gamma$ using cell outline data}

We first apply the inference framework to the moss {\it Physcomitrium patens} protonemata  and {\it Medicago truncatula} root hairs that grow their tips approximately self-similarly. Since only the turgid cell outlines that are available to us at the moment \cite{tension-distrubution-chelladurai,dumais-2004}, we use the cell outlines in the turgid state to approximate the unturgid state ${R}({S})$, and solve the turgid state $r(s)$ using our model. 


\subsubsection{Fitting the data into a canonical shape}
The cell shape data were provided to us as a set of outlines, each given by a sequence of $X$ and $Y$ coordinates of points on the cell outline. We first truncated each sequence of $X$ and $Y$ coordinates to remove any region beyond the point where the shank width stops increasing. 



Our initial attempt to characterize the outline coordinates is through a one-parameter family of shapes comprising a cylindrical pipe with ellipsoidal endcap \cite{Fayant2010}. However, this family of shapes does not fit the available data (moss {\it Physcomitrium patens} protonemata and {\it Medicago truncatula} root hairs) well in the least square sense. See Fig. \ref{fig:dumais-point-cloud-fitting} for example. 

\begin{figure}[ht]
        \includegraphics[height=7cm]{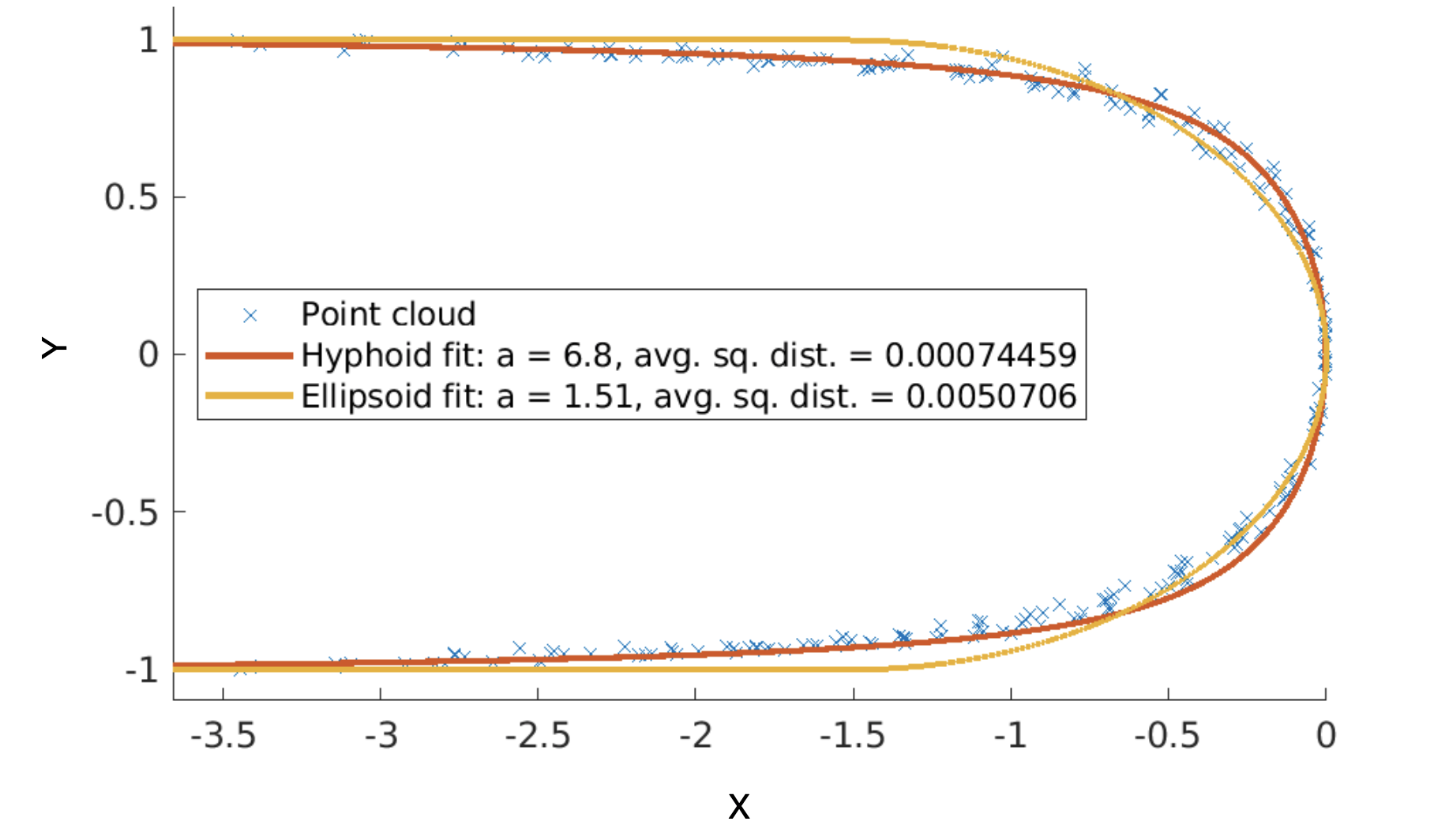}
    \caption{Fitting the point cloud of root hair shapes with hyphoid and ellipsoid shapes. The fitting is done by minimizing the sum of squared distances between the points and the fit curve.}
    \label{fig:dumais-point-cloud-fitting}
\end{figure}
Then, we adopt a one-parameter family of cell outline shapes given by so-called hyphoid shapes \cite{BartnickiGarcia1989}. These shapes are defined by the equation
\begin{eqnarray} \label{eqn:hyphoid-shape}
    z = \frac{\pi r}{a} \cot(\pi r).
\end{eqnarray}
Here $a$ is a shape parameter controlling the shape of the tip; larger values of $a$ correspond to a flatter tip, and in the limit as $a \to \infty$, the hyphoid shape approaches a rectangle.




We begin by using material parameters $K_h = \mu_h = 24$ such that the turgid and unturgid states are not different drastically (the choice of $K_h$ and $\mu_h$ will be discussed at the end of this subsection). See Fig.~\ref{fig:2} for the cell shapes from different cell types. Then we infer their exocytosis distribution $\gamma$ using Eq.(\ref{eqn:self-similarity-condition}). For moss protonema canonical shapes, we find that exocytosis is highest at the tip (e.g., see Fig.~\ref{fig:2}a for chloronema). 
\begin{figure}[!b]
            \includegraphics[width=0.9 \linewidth]{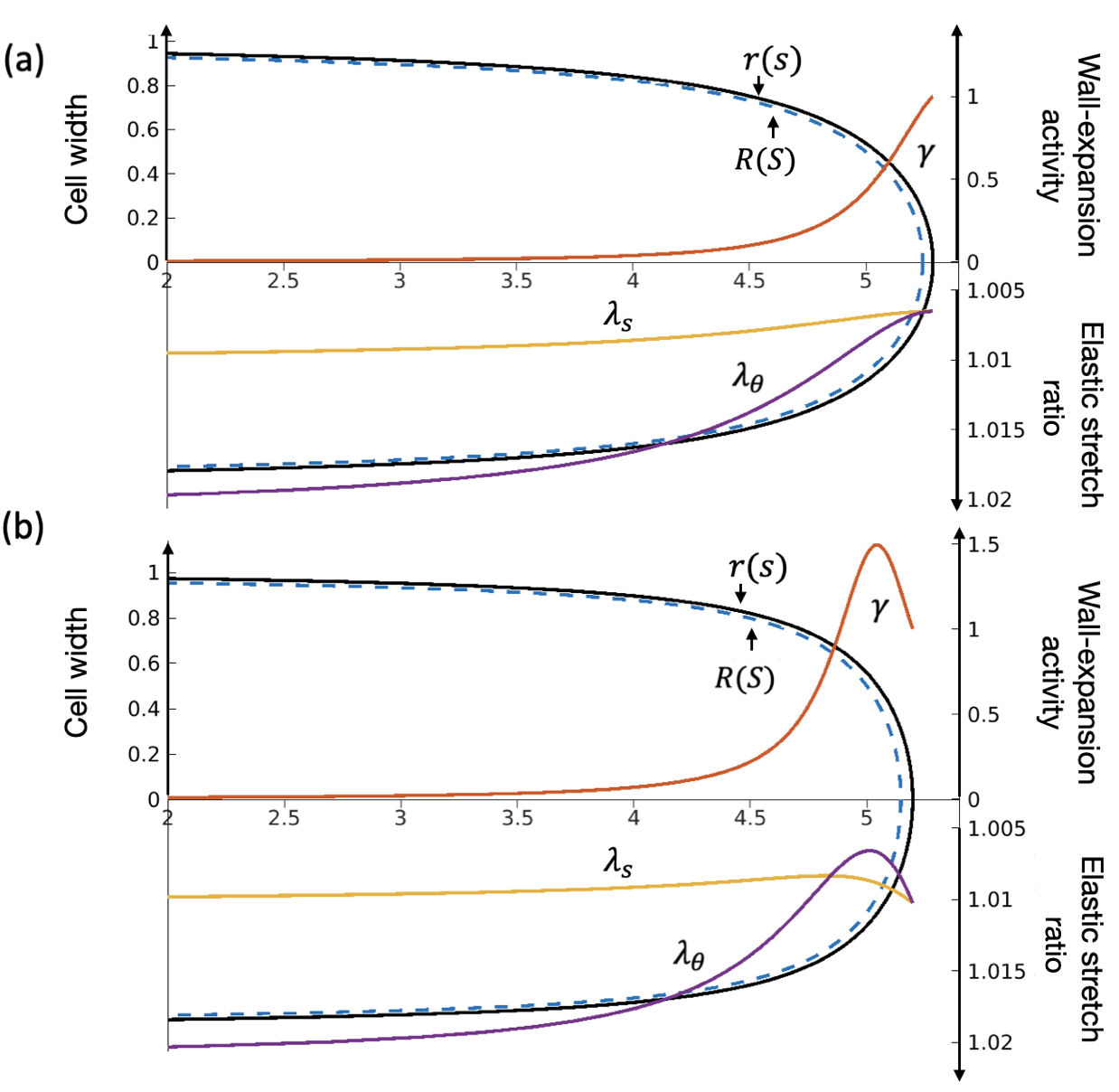}
    \caption{Cell outlines for the moss chloronema and root hair, with distributions of elastic stretch ratios and the exocytosis profile. Dashed outlines are for the unturgid states, and solid outlines are for the turgid states. (a) For the canonical shape of moss chloronema  ($a=3.4$), the exocytosis profile $\gamma$ is monotonically decreasing from the tip.  (b) For the canonical shape of {\it Medicago truncatula } root hair ($a=6.8$), $\gamma$ is non-monotone and has its global maximum away from the tip. We note they have distinct elastic-stretch distributions: $\lambda_s$ and $\lambda_\theta$. The elastic stretch ratios are solved with $K_h = \mu_h = 24$. See Fig.\ref{fig:more}  for the robustness of the results in changing wall material properties.}
        \label{fig:2}
\end{figure}

Even without numerically computing $\gamma$ from Eq.~(\ref{eqn:self-similarity-condition}), this result can be quickly deduced from Theorem 1. Specifically, since $\lambda_\theta'(0) > \lambda_s'(0) > 0$ (Fig.~2a), it follows that $\lambda_s'(0) - 2\lambda_\theta'(0) < 0$, indicating that $\gamma(0)$ has a local maximum at the tip.

In contrast, the canonical root hair shape exhibits exocytosis peaking annularly around the tip (Fig.~\ref{fig:2}b). In this case, the elastic stretch ratios satisfy $\lambda_\theta'(0) < \lambda_s'(0) < 0$, leading to $\lambda_s'(0) - 2\lambda_\theta'(0) > 0$. According to Theorem 1, this implies that $\gamma(0)$ is a local minimum. Consequently, if a global maximum exists, it must occur elsewhere. Numerical computation of the $\gamma$ profile confirms this, showing that $\gamma$ peaks at some ${S}_c > 0$ (Fig.~\ref{fig:2}b).

One may question whether the results are sensitive to changes in the elastic moduli \(K_h\) and \(\mu_h\). We find that increasing these moduli reduces the magnitude of elastic stretch ratios  \(\lambda_s\) and \(\lambda_\theta\), while decreasing them leads to larger elastic stretch ratios, without significantly altering their spatial distributions. Figure~\ref{fig:more}a illustrates this trend for \(K_h = \mu_h = 6\), where the maximal elastic stretch ratio  is approximately \(1.08\) (corresponding to an elastic strain of \(0.08\)), compared to the case of \(K_h = \mu_h = 24\), where the maximal stretch ratio  is about \(1.02\) (elastic strain of \(0.02\)). This corresponds to a four-fold reduction in elastic strain as the elastic moduli increase by a factor of four, consistent with expectations from linear elasticity.

For small strains, the nonlinear elasticity equations (\ref{rel1}), (\ref{rel2}) become equivalent to the linear elasticity equation (\ref{linear_law}), where the elastic moduli relate to Young’s modulus \(E_h\) and Poisson’s ratio \(\nu\) as $ 
\mu_h = \frac{E_h}{2(1+\nu)}, K_h = \frac{E_h}{2(1-\nu)}$. When \(K_h = \mu_h\), this formulation implies \(\nu = 0\), meaning there is no lateral contraction upon stretching. We also examine cases with \(\nu > 0\) and find that the elastic-stretch distribution remains qualitatively unchanged (Fig.~\ref{fig:more}b).

Since the inference of \(\gamma\) depends on cell shape and elastic strains, and the latter maintains its spatial distribution across different elastic moduli, the inferred values of \(\gamma\) remain qualitatively the same. In summary, we have demonstrated the robustness of the inference against variations in the wall material properties.

\subsection{Inference validation and connection to the fluid model}
To validate that the inferred $\gamma$ indeed contributes to the self-similar growth of the cell with the corresponding shape, we input the inferred $\gamma$ from Fig.~\ref{fig:2}b into the full forward model, Eqs.~(\ref{eq1})–(\ref{bc2}), and demonstrate that the elongation remains self-similar (Fig.~\ref{fig:25}a), preserving the cell shape shown in Fig.~\ref{fig:2}b.
\begin{figure}[!b]
            \includegraphics[width=0.9 \linewidth]{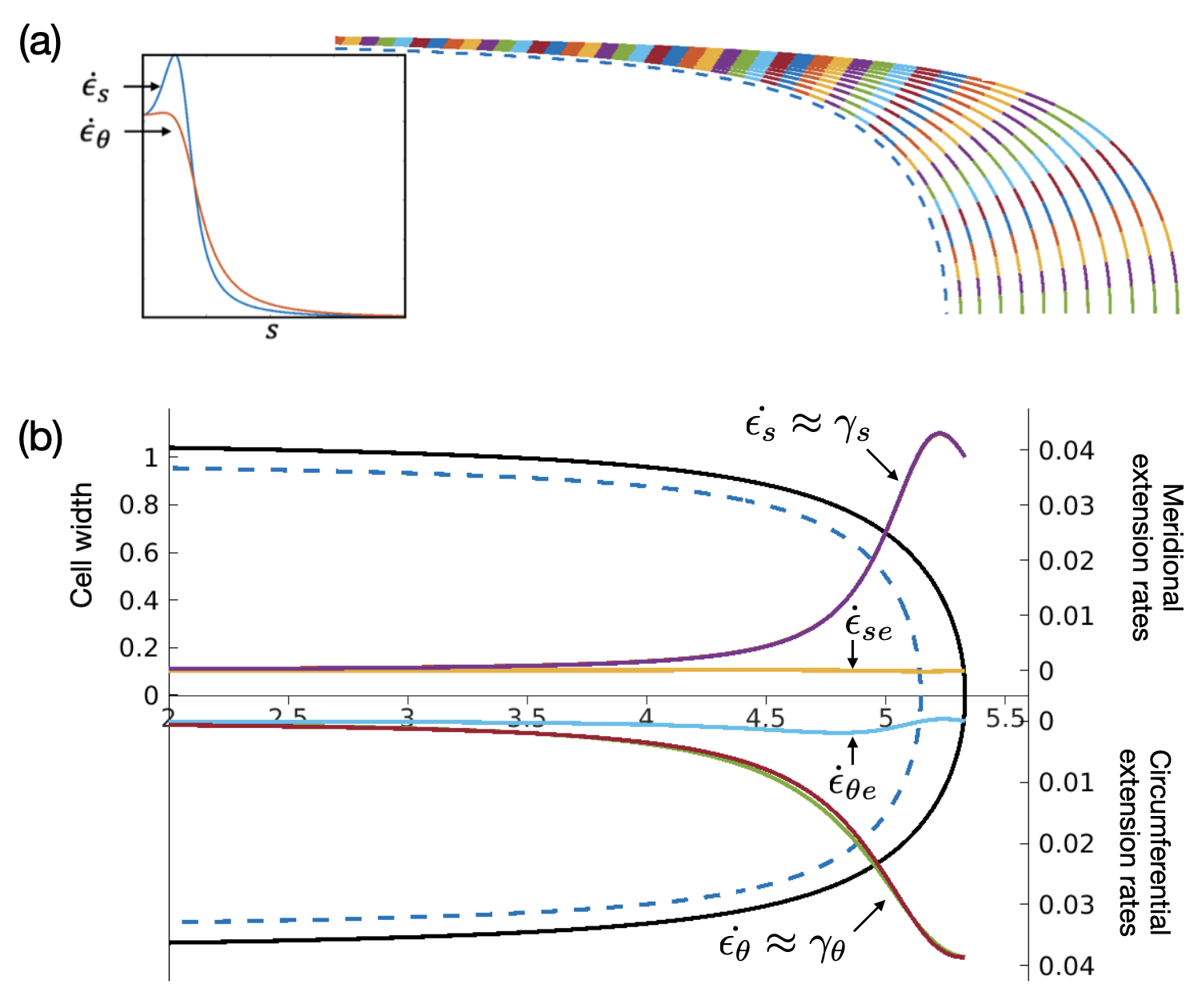}
    \caption{(a) Simulation of self-similar growth of the root-hair elongation by inputting the inferred $\gamma$ from Fig.\ref{fig:2}b into the forward model ($K_h = \mu_h = 24$). The predicted observable strain rates $\dot\epsilon_s$ and $\dot\epsilon_\theta$ of {\it Medicago truncatula } root hair (inset) are computed from the simulation and they are qualitatively consistent with the measurements in \cite{dumais-2004}. (b) Cell outlines for the root hair with $K_h = \mu_h = 6$, with distributions of  strain rates. Decompositions of $\dot{\epsilon}_s = \dot{\epsilon}_{s e} + \gamma_s$ (red, yellow, purple) and $\dot{\epsilon}_\theta = \dot{\epsilon}_{\theta e} + \gamma_\theta$ (green, cyan, maroon). See text for more details.}
        \label{fig:25}
\end{figure}

Then, we compare the model-predicted cell wall strain rates, $\dot{\epsilon}_s$ and $\dot{\epsilon}_\theta$, with experimental measurements from Fig.~4c–f in \cite{dumais-2004}. Their computation uses Eq.(\ref{dotrate}).  
The results exhibit qualitative agreement, capturing the non-monotonic behavior of both $\dot{\epsilon}_s$ and $\dot{\epsilon}_\theta$ (Fig.~\ref{fig:25}a, inset).

In the self-similar growth regimes, one can derive Eqs.(\ref{decomposition_s}) and (\ref{decomposition_r}) which decompose  $\dot{\epsilon}_s$ and $\dot{\epsilon}_\theta$ additively to the part due to irreversible extension $\gamma_s$ and $\gamma_\theta$, and that due to elastic deformation ${\dot{\epsilon}}_{se}$ and ${\dot{\epsilon}}_{\theta e}$.  $\dot{\epsilon}_{se}$ and $\dot{\epsilon}_{\theta e}$ come from the moving material traveling along the cell wall outline that is under different elastic stretch and can be removed when the turgor pressure is removed. 

Interestingly, we find the contribution of $\dot{\epsilon}_{se}$ and $\dot{\epsilon}_{\theta e}$ is negligible for a wide range of material parameters: $6\leq K_h=\mu_h\leq 24$. See Fig.~\ref{fig:25}b for $K_h=\mu_h=6$. Notice even when the cell wall is soft such that the turgid cell outline is significantly expanded than the unturgid one, the strain rates due to elastic deformation ${\dot{\epsilon}}_{se}$ and ${\dot{\epsilon}}_{\theta e}$ are still close to $0$.  This observation gives the leading-order approximation $\dot{\epsilon}_s \approx \gamma_s = \beta \gamma (\lambda_s-1)$ and $\dot{\epsilon}_\theta \approx \gamma_\theta=\beta \gamma (\lambda_\theta-1)$. When the elastic strains $(\lambda_s-1)$ and $(\lambda_\theta-1)$ are small, the nonlinear elasticity model Eqs.(\ref{rel1}) and (\ref{rel2}) converges to the linear elasticity Eq.(\ref{linear_law}). This together gives $\dot{\epsilon}_s \approx \gamma_s = \frac{\beta \gamma}{E_h} ({\sigma_s - \nu \sigma_\theta})$ and  $\dot{\epsilon}_\theta \approx \gamma_\theta=\frac{\beta \gamma}{E_h}(\sigma_\theta - \nu \sigma_s)$, which is the viscous-fluid model in \cite{rojas-pollen-tube,ohairwe2024fitness}. This means the growth-elasticity model is equivalent to the viscous-fluid model at the leading order.

\subsection{Connecting cell morphology with $\gamma$}
Further, we apply our inference framework to theoretically link the exocytosis profile with cell shapes across a range of morphologies. We focus on a family of cell shapes that can be fitted by the hyphoid function
$z = \left(\frac{\pi r}{a}\right) \cot(\pi r)$,   
which describes structures such as fungal hyphae, root hair tips, and moss protonemata. By varying the parameter $a$, the cell shapes transition from tapered to flat apical domains (Fig.~\ref{fig:S3}a). We quantify this shape transition using the parameter $R_{\kappa}$, defined as the ratio of the tip radius of curvature to the turgid cell radius, $r(s^b)$, at the shank boundary. We investigate this relationship over the range $R_{\kappa} \in [0.55, 1.2]$.

As the apical domain becomes flatter, the exocytosis profile $\gamma$ transitions from a monotonically decreasing function of the meridional distance to a non-monotonic profile (Fig.~\ref{fig:S3}d). The strain rates, $\dot{\epsilon}_s$ and $\dot{\epsilon}_\theta$, exhibit similar behavior to the exocytosis profile (Fig.~\ref{fig:S3}e,f). In contrast, the elastic strains show a similar transition from monotonic to non-monotonic behavior but are spatially anticorrelated with the exocytosis and extension profiles (Fig.~\ref{fig:S3}b,c). We note that previous theoretical results have connected the geometry of the cell with wall extensibility or secretion profile numerically. In particular, Fig. 8 in \cite{dumais-2006} and Fig. 3 in \cite{fission-yeast} have shown that the spatial profile of the meridional curvature reflects the distribution of wall extensibility and the secretion profile, respectively. Here we include how the shape transition shifts the distribution of the curvatures in Fig.~\ref{fig:S3}g,h. Note that as cells become flatter at the tip, the meridional curvature distribution also shift in a way reflecting the shift of exocytosis profile, while circumferential curvature distribution is always monotone. 

 \begin{figure}[h!]
	            \includegraphics[width=1.0 \linewidth]{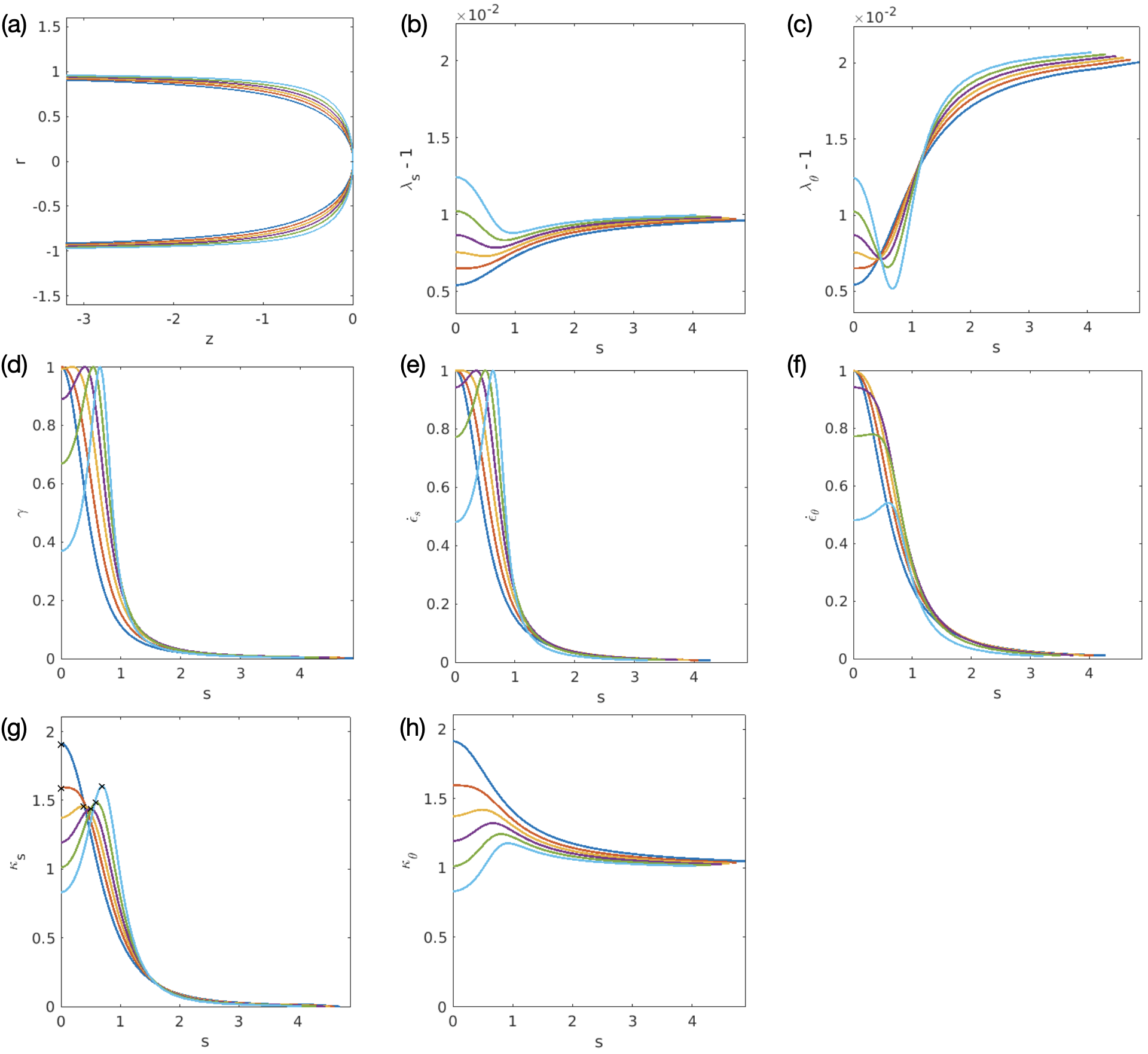}
      	\caption{(a) Sample cell outlines represented by $z=({\pi r}/{a})\cot(\pi r)$  showing a transition from tapered to flat shapes as $a$ increases, with values  $a = 3.4, 4.1, 4.8, 5.6, 6.8, 8.8$. Using the same color codes, we show the corresponding (b,c) elastic strains; (d) exocytosis profiles; (e,f) strain rates; (g,h) meridional and circumferential curvatures. }
	  \label{fig:S3}
\end{figure}

We quantify the relationship between the exocytosis profile and cell shape by measuring two key parameters as  $R_{\kappa}$ varies (see Fig.~\ref{Fig_sum}):  
1) the length of the exocytosis window, $w$, defined as the shortest arclength interval $[s_1, s_2]$ containing half of the total exocytosis window; and  
2) the meridional position, $l$, of the highest exocytosis level. 

To reduce uncertainty due to cell wall material properties, we conducted this study using different parameter values in both the nonlinear elasticity model and the linear elasticity model (see legend in Fig.~\ref{Fig_sum}b). All parameter sets yielded consistent results. For relatively tapered cells ($R_{\kappa} \leq 0.65$), changes in $R$ are primarily driven by the narrowing or widening of the exocytosis profile concentrated at the tip. This is reflected in the increase of $w$ while $l$ remains zero within this range (Fig.~\ref{Fig_sum}a). This regime also includes cell shapes characteristic of moss chloronema and caulonema.  

Beyond the critical value of $R_{\kappa} \sim 0.65$, the exocytosis peak must shift away from the tip to sustain a flatter apical domain. This transition is evident from the increase in $l$ for $R_{\kappa} \geq 0.65$ (Fig.~\ref{Fig_sum}b). In this range, the canonical shape of {\it Medicago truncatula} root hairs (at $R_{\kappa} \approx 1$) corresponds to the widest secretion window ($w \approx 0.5$, Fig.~\ref{Fig_sum}a). As $R_{\kappa}$ increases further, $w$ decreases, suggesting that maintaining an extremely flat tip shape requires exocytosis to be confined to a narrow annular region—an arrangement that requires much more precise patterning of vesicle delivery to the membrane compared to wider annular zones or simpler tip-focused exocytosis.

\begin{figure}[h!]
    \centering
    \includegraphics[width=0.8\linewidth]{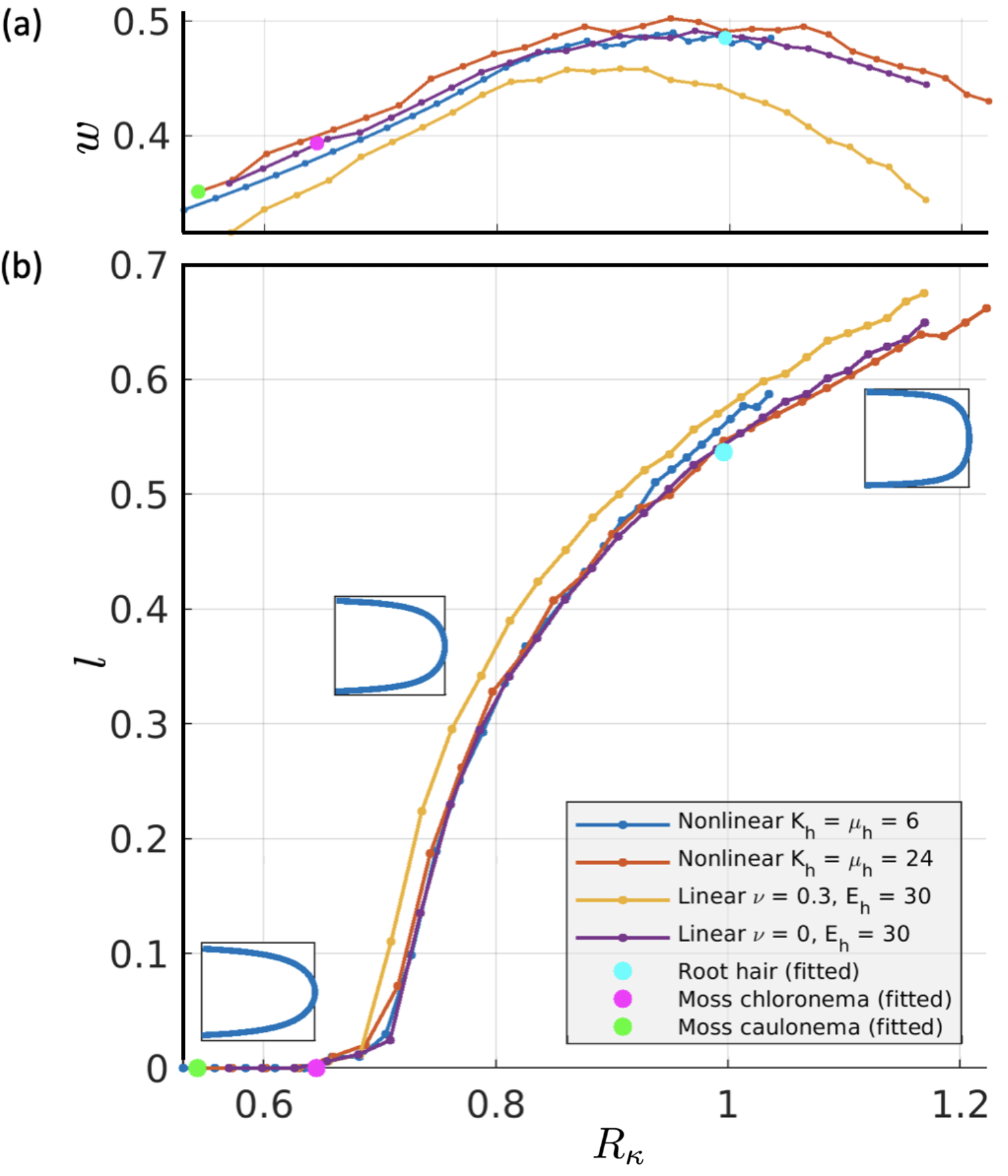}

    \caption{The width $w$ of the exocytosis profile (a) and the meridional position $l$ of secretion maximum (b) vs. $R_{\kappa}$, ratio of the radius of the tip curvature to the cell radius. The dependence of $w$ and $l$ on $R_{\kappa}$ is insensitive to the change of material parameters and the change of constitutive laws (see the legend for the correspondence of each curve). On each curve, data points are from inferences using unturgid cell shapes with values of $a$ ranging from 3.4 to 9, in increments of 0.2.
}
  \label{Fig_sum}\end{figure}

\section{Discussion}
\subsection{In comparison with previous models}
\label{comparision_sec}
In this paper, we introduced a dual-configuration model inspired by recent experiments on deflating fungal cells \cite{davi-2019,chevalier2023cell}, which enable the measurement of cell wall mechanics involving both turgid and unturgid configurations. This model separates the wall extension process into two components: an experimentally trackable elastic deformation and the irreversible extension of the underlying unturgid configuration.

Although this modeling approach appears more complex than previous models, it shares strong connections with them. First, regarding turgor-induced wall tension mechanics, our vector equation, Eq.~(\ref{eqn:local-force-balance}), which explicitly relies on the local cell radius and meridional tangent directions, is equivalent to the classical curvature-based equations $\kappa_s \sigma_s + \kappa_\theta \sigma_\theta = P$ and $\kappa_\theta \sigma_s = P/2$.

Second, the dual-configuration model captures wall fluidity, a key feature in previous models. We categorize earlier models into three groups:  
1) Viscous (or viscoplastic) fluid models \cite{rojas-pollen-tube,deJong2020,Camps2009,dumais-2006,Camps2012,DrakeVavylonis,ohairwe2024fitness},  
2) Dual-flow model \cite{fission-yeast}, and  
3) Discrete elastic models \cite{goriely-tabor-jtb,self-similar-tip-growth-goriely-tabor,Fayant2010}.  

We have shown that our dual-configuration model is more general but reduces to the dual-flow model (2) in the self-similar growth regime. Additionally, in this regime, it is equivalent—at leading order—to fluid models (1), which rely on the relations $\dot{\epsilon}_s \propto \gamma (\sigma_s - \nu \sigma_\theta)$ and $\dot{\epsilon}_\theta \propto \gamma (\sigma_\theta - \nu \sigma_s)$ (e.g., \cite{rojas-pollen-tube,ohairwe2024fitness}). Comparisons with discrete elastic models (3) are more challenging since these models describe wall extension as a sequence of discrete events alternating between cell inflation due to turgor pressure and a reset of the inflated configuration to a new, unextended state.

\subsection{On the connection between exocytosis and cell shape}
Previous fluid models have explored the connection between chemical profiles driving wall extension and cell morphologies in the self-similar growth regime, often relying on forward simulations where cell shape is predicted from a given exocytosis profile \cite{fission-yeast,ohairwe2024fitness}. Notably, \cite{ohairwe2024fitness} identified bifurcations where the same exocytosis profile yields multiple cell shapes.

In contrast, our inference framework works in reverse, showing that cell wall geometry and mechanics uniquely determine the exocytosis profile, as the problem reduces to an initial value problem of a first-order ordinary differential equation. Thus, while one exocytosis profile may correspond to multiple morphologies, two distinct profiles cannot produce the same morphology given that the elastic strains of the cell wall are uniquely determined by the specified elastic moduli and wall thickness.

Our framework can be applied directly to cell wall configuration data without observing wall flow, as the deformation relationship between turgid and unturgid configurations is sufficient. Ideally, this relationship can be experimentally obtained by tracking material points between configurations. In the absence of such data, we use the cell outline data in the turgid state, and generate deformation fields by solving the elastostatic system Eqs.~(\ref{stretches})–(\ref{bc2}), and apply the self-similarity condition Eq.~(\ref{eqn:self-similarity-condition}).

From tapered to flattened tip configurations, we identify two strategies for achieving specific cell shapes: adjusting the polarity of the exocytosis profile and relocating the exocytosis ``hotspot." For tapered shapes resembling moss protonemata ($R_{\kappa} < 0.65$), changes in the secretion window width can modify the tip shape. In contrast, for flatter-tip shapes ($R_{\kappa} > 0.65$), the secretion hotspot must shift away from the tip to maintain the desired morphology.

Interestingly, the canonical shape of \textit{Medicago truncatula} root hairs (\(R_{\kappa} \approx 1\)) corresponds to the widest secretion window in our model. This suggests that such shapes may be more robustly maintained than flatter ones with \(R_{\kappa} > 1\), which require a narrower annular secretion zone (with $w$ decreasing from \(R_{\kappa} \approx 1\) in Fig.\ref{Fig_sum}a). While we do not question the viability of tapered geometries (\(R_{\kappa} < 1\))---which are widespread across fungal and plant species \cite{Camps2012}---we note that achieving \textit{flatter} shapes requires more spatially confined exocytosis. This increased localization makes shape maintenance more sensitive to small displacements of the secretion hotspot, as reflected in the larger values of \(|dR_{\kappa}/dl|\) in Fig.\ref{Fig_sum}b. Such sensitivity may help explain the relative rarity of extremely flat tip shapes across walled cell species.

Our findings also shed light on the ongoing debate regarding the spatial distribution of exocytosis in tip-growing cells, which remains unclear due to experimental limitations \cite{luo2016measuring}. While exocytosis is often assumed to concentrate at the tip, vesicle distribution measurements in pollen tubes suggest that an annular exocytosis pattern around the tip can also occur \cite{zonia2008vesicle,bove2008magnitude}.

Within the spatially uniform elastic moduli regime, we have shown our result is robust to the change of the cell wall material property, and in particular shows that the self-similar growth does not require the wall material property to have a spatial gradient as suggested in previous works \cite{goriely-tabor-jtb,self-similar-tip-growth-goriely-tabor,Fayant2010}. 

\subsection{Conclusion}
In conclusion, we have introduced a new dual-configuration model of tip growth, establishing theoretical connections with previous fluid models. In the self-similar growth regime, we provided a proof of concept for a novel approach to linking the exocytosis profile with cell morphology by tracking elastic deformations between turgid and unturgid configurations. Using this framework, we have connected secretion profiles with cell morphology in moss protonemata, {\it Medicago truncatula} root hairs, and other theoretical shapes within the family of hypha-like forms.

\section*{Acknowledgment}
We thank Dr. Jacques Dumais for providing the root hair cell outline data from \cite{dumais-2004}. M.W. and K.S. acknowledge partial funding from the National Science Foundation- Division of Mathematical Sciences (NSF-DMS) grant DMS-2012330 and 2144372 (CAREER). L.V. acknowledges partial funding from the Division of Molecular and Cellular Biosciences (NSF-MCB) grant MCB-1253444 and 2154573. M.W. gratefully acknowledges the generous hosting and support provided by the Center for Computational Biology at the Flatiron Institute during her sabbatical.

\appendix

\section{The flow and strain rates from the dual-configuration model}
\label{Appendix_B}
Equations~\ref{dotrate}, \ref{eq1}, and \ref{stretches} define the dual-configuration framework and imply a connection between the meridional wall-surface flow velocities in the turgid and unturgid geometries, denoted \( v_s \) and \( V_S \), respectively. These velocities can be related by differentiating the arc-length mappings \( s(s_p, t) \) and \( S(s_p, t) \) with respect to time.

Specifically, we note that \( s(s_p, t) \) can be expressed as a function of \( S \) through a composite mapping: 
\[
s(s_p, t) = \tilde{s}(S(s_p, t), t),
\]
where \( \tilde{s}(S, t) \) denotes the turgid arc-length expressed in terms of the unturgid coordinate \( S \). Applying the chain rule yields:
\[
\frac{\partial s}{\partial t}\bigg|_{s_p} = \left.\frac{\partial \tilde{s}}{\partial t}\right|_{S} + \frac{\partial \tilde{s}}{\partial S} \cdot \frac{\partial S}{\partial t}\bigg|_{s_p}.
\]
By a slight abuse of notation—writing \( \tilde{s} \) as \( s \) —we obtain:
\begin{eqnarray}
v_s = \frac{\partial s}{\partial t}\bigg|_{s_p} = \left.\frac{\partial s}{\partial t}\right|_{S} + \lambda_s V_S,
\label{flow}
\end{eqnarray}
where \( \lambda_s = \partial s / \partial S \) is the meridional stretch ratio. This is Equation~\ref{flow}, which reflects how the velocity in the deformed (turgid) configuration is composed of both the time evolution at fixed reference position and the underlying material flow.

In case of self-similar growth, where the cell shape reaches a steady state,  we have $\frac{\dee s}{\dee t}|_{{S}} = 0$, which gives $v_s = \frac{\dee s}{\dee {S}} V_{S} = \lambda_s V_{S}$ from \cite{fission-yeast}. Given its formulation based on flow descriptions, we denote the theory from \cite{fission-yeast} as dual-flow model. 

Restricted to the self-similar growth regime, \cite{fission-yeast} also derived the additive decomposition among different strain rates. We include it here as it is also a property of the two-configuration model in the self-similar growth regime. By taking $\frac{\dee }{\dee s}$ of the equation $v_s = \lambda_s V_{S}$ and using  $\lambda_s = \frac{\partial s}{\partial S}$, we obtain \[\frac{\dee v_s}{\dee s} = \frac{\dee \lambda_s}{\dee s} \frac{v_s}{\lambda_s} + \frac{\dee V_{S}}{\dee {S}}.\]

Recognizing $\frac{\dee v_s}{\dee s} = \dot{\epsilon}_s$ and $\frac{\dee V_{S}}{\dee {S}} = \gamma_s$  from Eqs. \ref{dotrate} and \ref{eq1}, we have thus obtained the additive decomposition of the meridional strain rate into two parts:
\begin{eqnarray} 
\label{decomposition_s}
    \dot{\epsilon}_s = \dot{\epsilon}_{se} + \gamma_s, 
\end{eqnarray} 
where $\dot{\epsilon}_{se} = \frac{\dee \lambda_s}{\dee s} \frac{v_s}{\lambda_s}$ is the rate of extension due to elastic deformation and $\gamma_s$ is meridional  growth strain rate in the unturgid configuration. In the circumferential direction, a similar argument yields the analogous decomposition of the circumferential strain rate as
\begin{eqnarray} 
\label{decomposition_r}
    \dot{\epsilon}_\theta = \dot{\epsilon}_{\theta e} + \gamma_\theta,
\end{eqnarray} 
where $\dot{\epsilon}_{\theta e} = \frac{\dee \lambda_\theta}{\dee s} \frac{v_s}{\lambda_\theta}$. We note that Eqs.~\ref{decomposition_s} and \ref{decomposition_r} are not valid in general, especially outside the self-similar growth regime.  

\section{Derivation of Equation \ref{eqn:local-force-balance}}
\label{Appendix_force}
We derive the force balance used in Eq.(\ref{eqn:local-force-balance}) by applying the surface divergence to the in-plane stress tensor \begin{eqnarray} 
\boldsymbol{\sigma} = \sigma_s \hat{t} \otimes \hat{t} + \sigma_\theta \hat{\theta} \otimes \hat{\theta},
\end{eqnarray}  on an axisymmetric shell to balance an internal pressure: $\nabla_s\cdot\boldsymbol\sigma +P\hat{n}=0$, 
where \( \sigma_s \) and \( \sigma_\theta \) are the meridional and circumferential tensions, respectively. The surface divergence operator in axisymmetric coordinates is
 \begin{eqnarray} 
\nabla_s = \hat{t} \frac{\partial}{\partial s} + \frac{\hat{\theta}}{r} \frac{\partial}{\partial \theta},
\end{eqnarray} 
where \( s \) is arclength along the meridian and \( r \) is the radial coordinate.

The divergence of the stress tensor then becomes
\begin{eqnarray} 
\nabla_s \cdot \boldsymbol{\sigma} 
= \nabla_s \cdot (\sigma_s \hat{t} \otimes \hat{t}) + \nabla_s \cdot (\sigma_\theta \hat{\theta} \otimes \hat{\theta}).
\end{eqnarray} 
The first term yields
\begin{eqnarray} 
\nabla_s \cdot (\sigma_s \hat{t} \otimes \hat{t}) 
= \frac{\partial (\sigma_s \hat{t})}{\partial s} + \frac{\sigma_s}{r} (\hat{r} \cdot \hat{t}) \hat{t}.
\end{eqnarray} 

The second term gives
\begin{eqnarray} 
\nabla_s \cdot (\sigma_\theta \hat{\theta} \otimes \hat{\theta}) 
= -\frac{\sigma_\theta}{r} \hat{r}.
\end{eqnarray} 

To obtain Eq.(\ref{eqn:local-force-balance}), we use the identity
\[
\frac{\sigma_s}{r} (\hat{r} \cdot \hat{t}) \hat{t}
= \frac{\sigma_s}{r} \hat{r}
- \frac{\sigma_s}{r} \hat{r}
+ \frac{\sigma_s}{r} (\hat{r} \cdot \hat{t}) \hat{t},
\]
and observe that the last two terms combine as
\[
- \frac{\sigma_s}{r} \hat{r}
+ \frac{\sigma_s}{r} (\hat{r} \cdot \hat{t}) \hat{t}
= - \frac{\sigma_s}{r} (\hat{r} \cdot \hat{n}) \hat{n}
= - \frac{\sigma_s \sin\alpha}{r} \hat{n},
\]
where \( \alpha \) is the angle from \( \hat{t} \) to \( \hat{r} \).

Thus, the total force balance becomes
\[
\vec{F} := \nabla_s \cdot \boldsymbol{\sigma} + P \hat{n} 
= \frac{\partial (\sigma_s \hat{t})}{\partial s} 
+ \frac{\sigma_s - \sigma_\theta}{r} \hat{r} 
- \frac{\sigma_s \sin\alpha}{r} \hat{n} 
+ P \hat{n} = 0,
\]
which is Eq.(\ref{eqn:local-force-balance}) in the main text.

\section{Solving the secretion profile $\gamma$ as an initial value problem}
\label{Appendix_A}
Let $U_s({S}) = \lambda_s({S}) - 1$, $U_\theta({S}) = \lambda_\theta({S}) - 1$, and $\tilde{U}({S}) = U_s({S}) / U_\theta({S})$. From the self-similarity condition Eq.(\ref{eqn:self-similarity-condition}) in the main text, and by assuming $U_s({S})>0$ and $U_\theta({S})>0$ we have
\begin{eqnarray}
\label{sec}
    \gamma({S}) = \frac{d{R} / d{S}}{U_\theta({S}) {R}({S})} \int_0^{{S}} U_s(w) \gamma(w)\ dw.
\end{eqnarray}
Differentiating through, we obtain
\begin{align}
\nonumber
    \gamma'({S}) &= \frac{d{R} / d{S}}{U_\theta({S}) {R}({S})} U_s({S}) \gamma({S}) + {\bigg (} \frac{U_\theta({S}) {R}({S}) d^2 {R} / d{S}^2 }{(U_\theta({S}) {R}({S}))^2}  \\
    &\qquad - \frac{d{R} / d{S} \left( \lambda_\theta'({S}) {R}({S}) + U_\theta({S}) d{R} / d{S} \right)}{(U_\theta({S}) {R}({S}))^2} {\bigg )} \int_0^{{S}} U_s(w) \gamma(w)\ dw \nonumber \\
    &= \tilde{U}({S}) \frac{d{R}/d{S}}{{R}({S})} \gamma({S}) + {\bigg (} \frac{d^2 {R} / d{S}^2}{U_\theta({S}) {R}({S})} \nonumber \\
    &\qquad - \frac{(d{R} / d{S}) (\lambda_\theta'({S}) {R}({S}) + U_\theta({S}) d{R} / d{S})}{(U_\theta({S}) {R}({S}))^2} {\bigg )} \frac{\gamma({S}) U_\theta({S}) {R}({S})}{d{R} / d{S}} \nonumber \\
    &= \left( (\tilde{U}({S}) - 1) \frac{d{R} / d{S}}{{R}({S})} + \frac{d^2 {R} / d {S}^2}{d{R} / d{S}} - \frac{\lambda_\theta'({S})}{U_\theta({S})} \right) \gamma({S}). 
\end{align}
Thus, we define
\begin{eqnarray}
\label{ff}
    f({S}) = (\tilde{U}({S}) - 1) \frac{d{R} / d{S}}{{R}({S})} + \frac{d^2 {R} / d {S}^2}{d{R} / d{S}} - \frac{\lambda_\theta'({S})}{U_\theta({S})},
\end{eqnarray}
and convert the self-similarity condition to an initial value problem with $\gamma'({S}) = f({S}) \gamma({S})$ and the initial condition $\gamma(0) = \bar{\gamma}$. Note that $f$ depends only on $\lambda_s$, $\lambda_\theta$, $\lambda_\theta'$, ${R}$, $d{R} / d{S}$, and $d^2 {R} / d{S}^2$, which means that the secretion is constrained by only the information of the strains up to first-order derivative and the shape up to second-order derivative.

Furthermore, one may verify that 
\begin{eqnarray}
    \lim_{{S} \to 0+} f({S}) = \frac{\lambda_s'(0) - 2 \lambda_\theta'(0)}{U_\theta(0)}
\end{eqnarray}
is finite as long as $U_\theta(0)=\lambda_s({0}) - 1\neq 0$ and the tip is smooth such that $d{R}/d{S}$ is differentiable and $d{R}/d{S}|_{{S}=0}=1$.  This can be seen from the following details: 1)     $\lim_{{S} \to 0+} (\tilde{U}({S}) - 1) \frac{d{R} / d{S}}{{R}({S})}  = \lim_{{S} \to 0+} d{R} / d{S} \frac{(\tilde{U}({S}) - 1)' }{d{R} / d{S}} = \lim_{{S} \to 0+} (\tilde{U}({S}) - 1)' = \frac{(U_s-U_\theta)'U_\theta-(U_s-U_\theta)(U_\theta)'}{(U_\theta)^2} =\frac{(U_s-U_\theta)'}{U_\theta}$. 2) $\lim_{{S} \to 0+} \frac{d^2 {R} / d ({S})^2}{d{R} / d{S}} = 0$ due to the fact that for a arc-length parameterized curve $\alpha({S}) = ({Z}({S}), {R}({S}))$, we must have $\alpha''({S})\perp \alpha'({S})$. Evaluating this relation at ${S}=0$ leads to $d^2 {R} / d {S}^2|_{{S}=0}=0$. 3) Thus we have    $ \lim_{{S} \to 0+} f({S}) = \frac{(U_s-U_\theta)'}{U_\theta} +0- \frac{\lambda_\theta'({S})}{U_\theta({S})} = \frac{\lambda_s'(0) - 2 \lambda_\theta'(0)}{U_\theta(0)}$.

We define $f(0)= \lim_{{S} \to 0+} f({S})$. Then there exists a unique solution $\gamma({S})=\overline{\gamma}\exp(\int_0^{{S}}fd\omega)$ on some interval $[0,\Sigma)$ as long as $f$ is continuous on $[0,\Sigma)$, which requires:  $d{R}/d{S} \neq 0$, $R\neq0$, and $\lambda_s \neq1$ in $(0,\Sigma)$.

The above analysis gives rise to criteria for detecting whether the distribution of secretion is monotonic (Theorem \ref{thm:montonicity-condition} ) and determining if there exists a local maximum of $\gamma$ at the tip (Theorem \ref{thm:local extrema} in the main text).

\begin{theorem} \label{thm:montonicity-condition}
   Given $\lambda_s$ continuous and $\lambda_\theta$, ${R}$ and $d{R}/d{S}$ continuously differentiable, the secretion distribution $\gamma$ from Eq.(\ref{sec}) and the condition $\gamma(0)=\overline{\gamma}>0$ is strictly monotonic on the interval $[0, \Sigma]\subset D = \{0\}\cup\{{S} | {R} \neq 0,d{R}/d{S} \neq 0,\lambda_s \neq1\}$ iff $f({S})$ from Eq.(\ref{ff})
   never vanishes away from the tip.
\end{theorem}

\begin{proof}
The forward direction is obvious using proof by contradiction. Now we prove the backward direction. Given the uniqueness property, the solution to $\gamma'({S}) = f({S}) \gamma({S})$ and $\gamma(0)=\overline{\gamma}$ takes the form $\gamma({S})=\overline{\gamma}\exp(\int_0^{{S}}fd\omega)$. If $f$ never vanishes, we have either $f>0$ or $f<0$. Thus, $\gamma({S})$ either monotonically increases (given  $f>0$) or monotonically decreases (given $f<0$).

\end{proof}

{\bf Proof of Theorem \ref{thm:local extrema}}.
\begin{proof}
We will only prove the first statement. When $ {\lambda_s'(0) - 2 \lambda_\theta'(0)}<0
$ and $\lambda_\theta(0)>1$, by continuity $\lim_{{S}\to0+}f({S})<0$.
Also by continuity, there exists $\delta>0$ such that $f({S})<0$ for $0\leq {S}<\delta$. By integrating $\gamma'=f\gamma$ , we obtain
   \begin{eqnarray}
      \gamma({S})-\gamma(0)= \int_0^{{S}}f\gamma d\omega<0
    \end{eqnarray}
    for any $0<{S}<\delta$. So $\gamma$ has one local maximum at ${S}=0$.
\end{proof}

\section{The cubic-spline solutions}
\label{Appendix_C}
The simulation of tip growth is determined by the elastostatic mechanical system Eqs. ~(\ref{stretches})–(\ref{bc2}) and the strain-promoted growth process Eqs.~(\ref{eq1}) and (\ref{strain-promoted}). At each time step $t$, we first solve the mechanical system Eqs. ~(\ref{stretches})–(\ref{bc2}) by the cubic-spline solutions of three mappings $r(s,t)$, $z(s,t)$, and $s({S},t)$ given the cubic splines of coordinates ${R}({S},t)$ and ${Z}({S},t)$. The steps of this procedure are summarized as follows. For simplicity of the presentation, we drop $t$.
\begin{enumerate}
    \item Parametrize the unturgid outline by $({Z}({S}), {R}({S}))$ and discretize this outline by $N$ marker points $(Z_i, R_i)$.  Interpolate between the marker points $(Z_i, R_i)$, using cubic splines.
    \item Finding the {\bf initial guess} of the steady-state turgid configuration $(z_i, r_i)$: approximate the steady-state turgid configuration $(z_i, r_i)$ by displacing the marker points $(Z_i, R_i)$ according to $d\vec{X}/dt =\vec F $ where $\vec X$ is the vector concatenated by the $z$- and $r$-coordinates of all marker points. The components of the force vector $\vec{F}$ along $z$- and $r$-direction on each marker point are computed from Eq.(\ref{eqn:local-force-balance}) in the main text. In Eq.(\ref{eqn:local-force-balance}), the tensions $\sigma_s^i$ and $\sigma_\theta^i$ are computed through Eqs.(\ref{rel1}),(\ref{rel2}) and (\ref{stretches}) in the main text where the elastic stretch ratios  $\lambda_s^i$ and $\lambda_\theta^i$ are computed based on the cubic splines of $s({S})$, $r(s)$, and ${R}({S})$, which are twice continuously differentiable (Notice $z(s)$ is constrained by $r(s)$ by the geometric fact $(dz/ds)^2+(dr/ds)^2=1$, so it can be constructed from $r(s)$ ). The components of $\vec{F}$ related to the boundary marker points ($i=1$ and $i=N$) rely on Eqs.(\ref{bc1}) and (\ref{bc2}) in the main text respectively.  The points $(z_i, r_i)$ are updated until the force residue $|\vec F|<\epsilon$. In this work, we choose $\epsilon<10^{-3}$. Notice that for each iteration, the cubic-spline coefficients of $s({S})$ and $r(s)$ are updated while ${R}({S})$ maintains the same.  
    \item Finding {\bf the solution of the mechanical equilibrium} turgid configuration $(z_i, r_i)$: based on the initial guess of $(z_i, r_i)$, use Newton's method to solve iteratively for the configuration of the points $(z_i, r_i)$ which leaves machine-zero force residue $\epsilon_{machine}$. Again, for each iteration, the cubic-spline coefficients are updated. Given both unturgid and turgid coordinates, we can solve elastic stretch ratios  $\lambda_s, \lambda_\theta$ by Eq.(\ref{stretches}) in the main text.
    
  \item For the {\bf cell wall extension}, we update the unturgid wall data $S_i$, $R_i$ for each marker point using the directional extension distributions according to Eqs.~(\ref{eq1}) and (\ref{strain-promoted}) by a linear approximation ${R}_{new} = {R}_{old} + \Delta t \gamma_\theta {R}_{old}$, and similarly for ${S}$, where $\Delta t$ is a small time step. We repeat step 1-4 at each time step.  
\end{enumerate}

For the inference of the exocytosis $\gamma$ over the cell outline, we only need to perform step 1-3 to solve for $\lambda_s, \lambda_\theta$, and then use directly Eq.(\ref{eqn:self-similarity-condition}) in the main text only requiring numerical integration.

\section{Supplemental figures}
Here we provide Fig.\ref{fig:more} in supplemental to Fig.\ref{fig:2}.
\begin{figure}[h!]
    \centering
            \includegraphics[width=0.9 \linewidth]{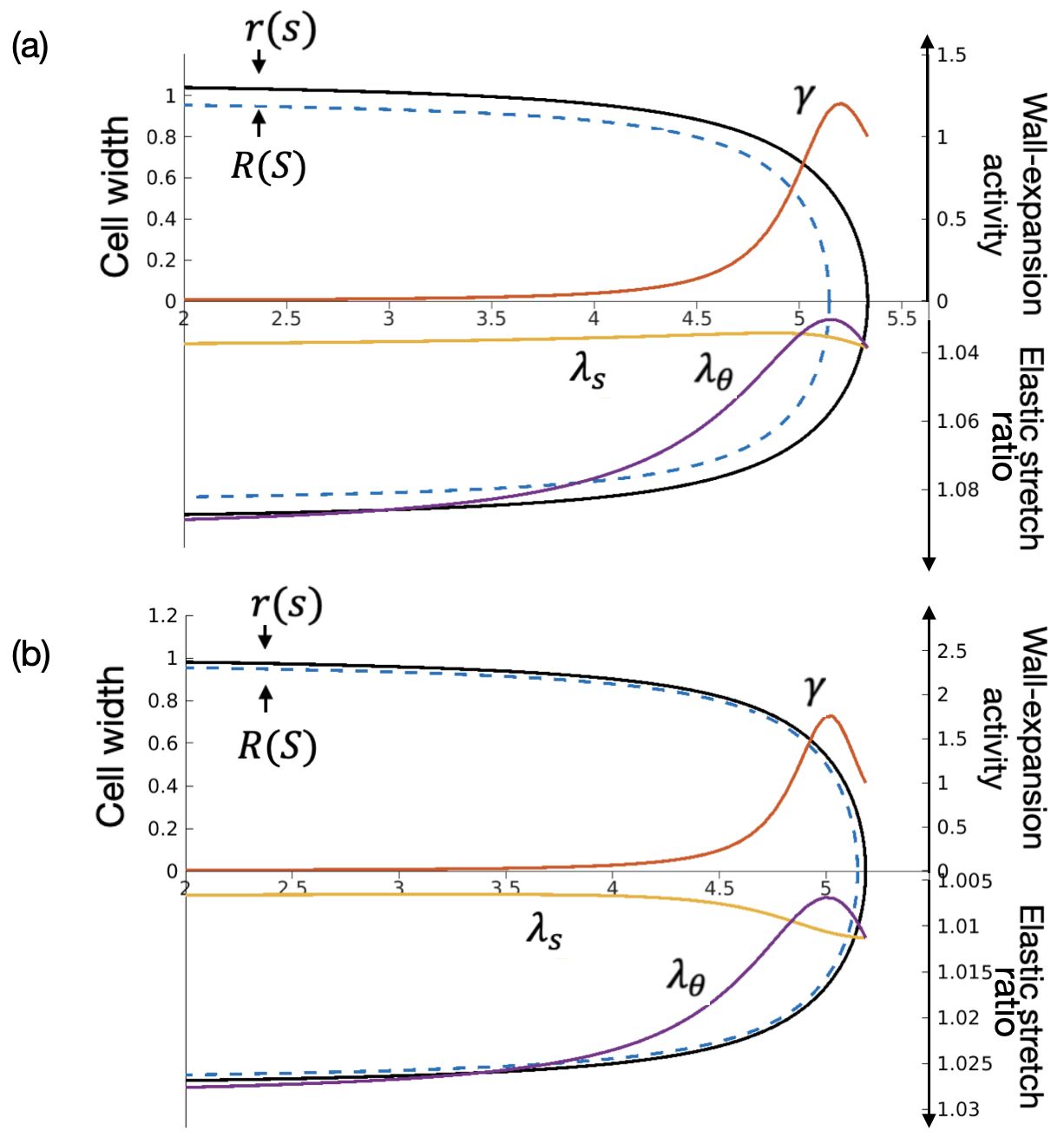}
    \caption{Sensitivity study of the inferred $\gamma$ given different material properties. (a) Deformation and inferred $\gamma$ for the same unturgid shape as Fig.\ref{fig:2}b in the main text, with softer material ($K_h = \mu_h = 6$).(b) Deformation and inferred $\gamma$ for the same unturgid shape as Fig.\ref{fig:2}b in the main text, under linear elasticity Eq.(\ref{linear_law})  with $\nu = 0.3$ and $E_h = 30$. }
        \label{fig:more}
\end{figure}


\bibliographystyle{elsarticle-harv}
\bibliography{refs}

\end{document}